\documentclass[journal, onecolumn]{IEEEtran}

\usepackage{url}
\usepackage{ifthen}
\usepackage{cite}
\usepackage[cmex10]{amsmath} 
\usepackage{bbold}
\usepackage{bookmark}

\usepackage{balance}
\usepackage{amsmath}
\usepackage{amsthm}
\usepackage{multirow}
\usepackage{amsfonts}
\usepackage{graphicx}
\usepackage{algorithm}
\usepackage{url}
\usepackage{lipsum}
\usepackage{booktabs,makecell,multirow}
\usepackage{tikz}
\usetikzlibrary{matrix}
\usepackage{rotating}
\usepackage{mathtools}
\usepackage{tkz-tab}

\usepackage{latexsym}
\usepackage{color}
\usepackage{graphicx}
\usepackage{threeparttable}
\usepackage{float}

\usepackage{algpseudocode}

\errorcontextlines\maxdimen

\newcounter{algsubstate}

\DeclareGraphicsExtensions{.tif,.pdf,.jpeg,.png,.jpg,.bmp,.svg,.eps,.EMF}

\newtheorem{theorem}{Theorem}
\newtheorem{definition}{Definition}
\newtheorem{lemma}{Lemma}

\newcommand{\cY}{\mathcal{Y}}
\newcommand{\cM}{\mathcal{M}}
\newcommand{\cA}{\mathcal{A}}

\newcommand{\cT}{\mathcal{T}}

\newcommand{\cB}{\mathcal{B}}
\newcommand{\cI}{\mathcal{I}}

\newcommand{\cH}{\mathcal{H}}

\newcommand{\bX}{\mathbf{x}}
\newcommand{\bY}{\mathbf{y}}

\newcommand{\bV}{\mathbf{v}}

\newcommand{\bR}{\mathbf{r}}

\newcommand{\bZ}{\mathbf{z}}

\newcommand{\bF}{\mathbf{f}}

\newcommand{\poly}{\mathsf{poly}}
\newcommand{\dec}{\mathrm{dec}}

\newcommand{\supp}{\mathsf{supp}}

\newcommand{\vecc}{\mathsf{vecc}}

\newcommand{\test}{\mathsf{test}}

\begin{document}

\title{Efficient designs for threshold group testing\\ without gap}

\author{Thach V. Bui~\IEEEmembership{Member, IEEE}, Yeow Meng Chee~\IEEEmembership{Fellow, IEEE}, and Van Khu Vu~\IEEEmembership{Member, IEEE}
\thanks{Thach V.~Bui, Yeow Meng Chee, and Van Khu Vu are with the Faculty of Engineering, National University of Singapore, Singapore. Email: \{bvthach,ymchee,isevvk\}@nus.edu.sg.}
}

\maketitle

\thispagestyle{plain}
\pagestyle{plain}


\IEEEpeerreviewmaketitle

\begin{abstract}
Given $d$ defective items in a population of $n$ items with $d \ll n$, in threshold group testing without gap, the outcome of a test on a subset of items is positive if the subset has at least $u$ defective items and negative otherwise, where $1 \leq u \leq d$. The basic goal of threshold group testing is to quickly identify the defective items via a small number of tests. In non-adaptive design, all tests are designed independently and can be performed in parallel. The decoding time in the non-adaptive state-of-the-art work is a polynomial of $(d/u)^u (d/(d-u))^{d - u}, d$, and $\log{n}$. In this work, we present a novel design that significantly reduces the number of tests and the decoding time to polynomials of $\min\{u^u, (d - u)^{d - u}\}, d$, and $\log{n}$. In particular, when $u$ is a constant, the number of tests and the decoding time are $O(d^3 (\log^2{n}) \log{(n/d)} )$ and $O\big(d^3 (\log^2{n}) \log{(n/d)} + d^2 (\log{n}) \log^3{(n/d)} \big)$, respectively. For a special case when $u = 2$, with non-adaptive design, the number of tests and the decoding time are $O(d^3 (\log{n}) \log{(n/d)} )$ and $O(d^2 (\log{n} + \log^2{(n/d)}) )$, respectively. Moreover, with 2-stage design, the number of tests and the decoding time are $O(d^2 \log^2{(n/d)} )$.
\end{abstract}

\section{Introduction}
\label{sec:intro}

The main objective of group testing is to find a small number of defective items in a large population of items~\cite{dorfman1943detection}. Defective items and non-defective items are defined by context. For example, in Covid-19 scenario, defective (respectively, non-defective) items are people who are positive (respectively, negative) to Coronavirus. In standard group testing (SGT), a pool of items is positive if there exists at least one defective item in the pool and negative otherwise. In some problems~\cite{malioutov2017learning,shin2009reactive}, there is a need for cooperation of at least two defective items in an effort to achieve a certain goal. A test on a set of items is positive if the certain goal is attained and negative otherwise. This means the SGT model does not fit to these problems. Damaschke~\cite{damaschke2006threshold} proposes a threshold group testing (TGT) model in which the outcome of a test on a subset of items is positive if the subset has at least $u$ defective items, negative if it has up to $\ell$ defective items, where $0 \leq \ell < u$, and arbitrary otherwise. The parameter $g = u - \ell - 1$ is called the gap. It is obvious that when $u = 1$, TGT reduces to SGT. When $\ell + 1 = u$, TGT is called \emph{TGT without gap}. 


In some problems~\cite{malioutov2017learning,shin2009reactive}, there is a need for cooperation of at least two defective items in an effort to achieve a certain goal. A test on a set of items is positive if the certain goal is attained and negative otherwise. This means the SGT model does not fit to these problems. A natural generalization of SGT is to introduce a threshold for making a test positive. Damaschke~\cite{damaschke2006threshold} proposes a threshold group testing (TGT) model in which the outcome of a test on a subset of items is positive if the subset has at least $u$ defective items, negative if it has up to $\ell$ defective items, where $0 \leq \ell < u$, and arbitrary otherwise. The parameter $g = u - \ell - 1$ is called the gap. It is obvious that when $u = 1$, TGT reduces to SGT. When $\ell + 1 = u$, TGT is called \emph{TGT without gap}. We illustrate TGT without gap for $u = 4$ versus standard group testing in Fig.~\ref{fig:ThresholdWGap}. The red and black dots represent defective and non-defective items, respectively. A subset containing defective items and/or non-defective items is a blue circle containing red and/or black dots. The outcome of a test on a subset of items is positive (\textcolor{red}{$+$}) or negative ({$-$}). In SGT (``Standard'' in the figure), the outcome of a test on a subset of items is positive if the subset has at least one red dot, and negative otherwise. In TGT without gap (``Threshold'' in the figure), the outcome of a test on a subset of items is positive if the subset has at least $u = 4$ red dots, negative if the subset has up to $u - 1 = 3$ red dots.

\begin{figure*}
\centering
\includegraphics[scale=0.4]{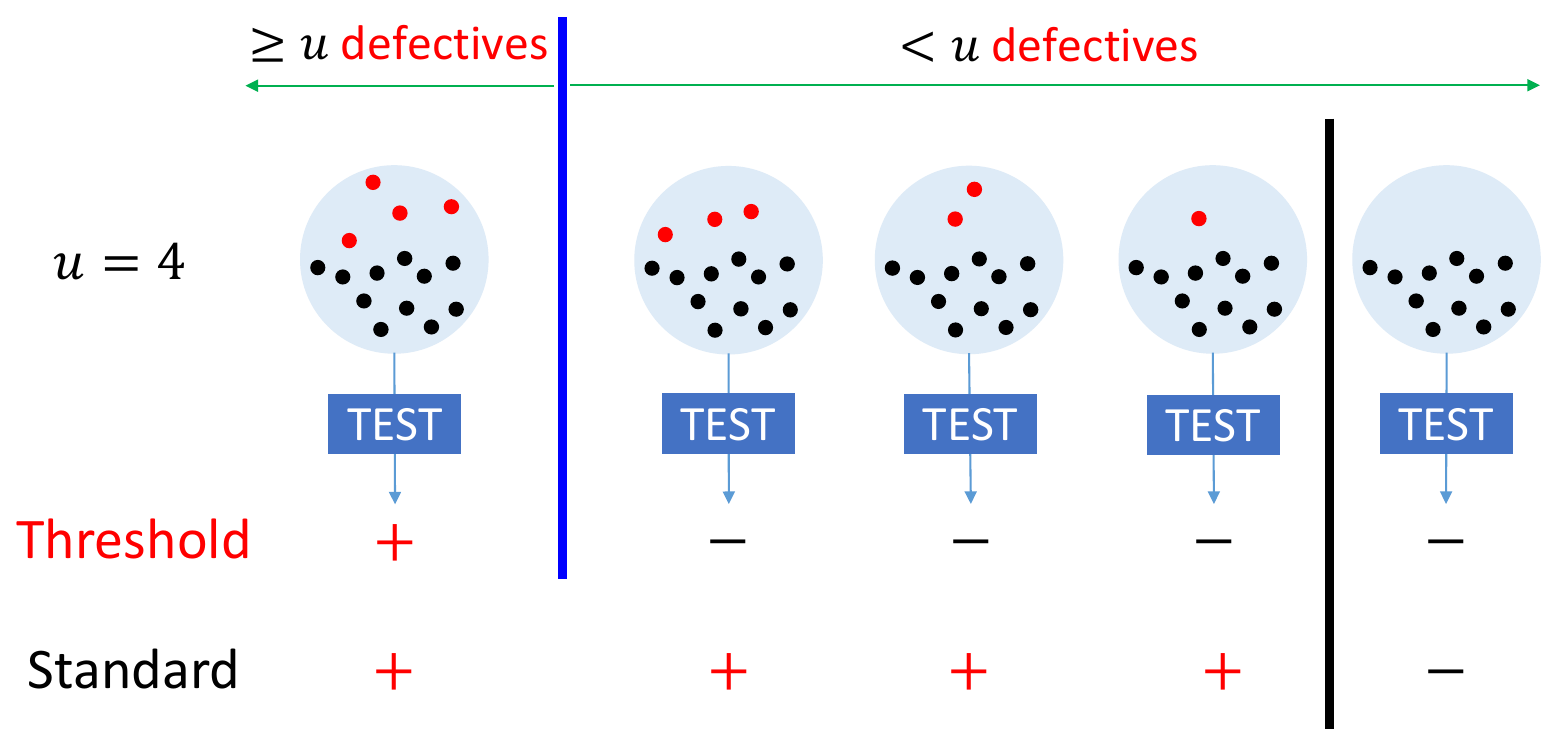}
\caption{Illustration of threshold group testing without gap for $u = 4$ versus standard group testing.}
\label{fig:ThresholdWGap}
\end{figure*}

There are many criteria in tackling group testing. Here we only mention four main criteria. The first criterion is the number of stages. When the number of stages is one, i.e., non-adaptive design, all tests are designed in priori and independently. Since the tests in this design can be performed simultaneously, it saves times, and therefore gives a significant advantage for some applications with a fast response demand such as identifying infected persons during pandemic~\cite{shental2020efficient}. Non-adaptive design has also been applied in other areas such as computational and molecular biology~\cite{du2000combinatorial}, networking~\cite{d2019separable}. However, the number of tests in non-adaptive design normally does not meet the information-theoretic bound that multi-stage designs can obtain. In multi-stage design, the design of tests in a stage depends on the design of the tests in the previous designs and therefore time consuming is essentially unavoidable. The second criterion is between deterministic and randomized design. Given the same input, a deterministic design produces the same result while a randomized one might not produce the same result because it includes a degree of randomness in it. The third criterion is the setting of the defective set. The combinatorial setting applies any configuration on the defective set under the predefined constrains while the probabilistic setting applies some distribution on the input items. Finally, both exact recovery and approximate recovery are also an important criterion studied in group testing.

In this work, we consider non-adaptive and 2-stage designs, deterministic design, combinatorial setting, and exact recovery. The formal problem definition is presented in the subsequent section.

\subsection{Problem definition}
\label{sub:pre:probDef}

We index the population of $n$ items from $1$ to $n$. Let $[n] = \{1, 2, \ldots, n \}$ and $P$ be the defective set with $|P| = d$. A test is defined by a subset of items $S \subseteq [n]$. A pool with a negative (positive) outcome is called a negative (positive) pool. The outcome of a test on a subset of items is positive if the subset contains at least $u$ defective items and is negative otherwise. Mathematically, the test outcome, denoted as $\test(S)$, is positive if $|S \cap P| \geq u$ and negative if $|S \cap P| < u$. This model is denoted as $(n, d, u)$-TGT. Note that $|P|$ can not be smaller than $u$ because all tests yield negative otherwise.

To represent tests in a stage, measurement matrices are often used. A $t \times n$ binary matrix $\cT =(t_{ij}) \in \{0, 1 \}^{t \times n}$ is defined as a measurement matrix, where $n$ is the number of items and $t$ is the number of tests, as follows. The $j$th item corresponds to the $j$th column of the matrix. An entry $t_{ij} = 1$ naturally means that item $j$ belongs to test $i$, and $t_{ij} = 0$ means otherwise. The outcome of all tests is $\bY = (y_1, \ldots, y_t)^T$, where $y_i=1$ if test $i$ is positive and $y_i = 0$ otherwise. An input vector $\bX = (x_1, \ldots, x_n) \in \{0, 1 \}^n$ represents $n$ items in which $x_j = 0$ if item $j$ is defective and $x_j = 0$ otherwise for $j \in [n]$. Let $\supp(\bV) = \{j \mid v_j \neq 0 \}$ be the support set for vector $\bV = (v_1, \ldots, v_w)$. The AND operator between two vectors of same size $\bZ = (z_1, \ldots, z_n)$ and $\bZ^\prime = (z_1^\prime, \ldots, z_n^\prime)$ is $\bZ \wedge \bZ^\prime = (z_1 \wedge z_1^\prime, \ldots, z_n \wedge z_n^\prime)$. Then the outcome vector $\bY$ is given by

\begin{alignat*}{3}
\bY := \begin{bmatrix}
y_1 \\
\vdots \\
y_t
\end{bmatrix} &:= &\cT \otimes \bX &:= \begin{bmatrix}
\cT(1, :) \otimes \bX \\
\vdots \\
\cT(t, :) \otimes \bX
\end{bmatrix} \\
&:= &\test(\cT, \supp(\bX)) &:= \begin{bmatrix}
\test(\cT(1, :) \wedge \bX) \\
\vdots \\
\test(\cT(t, :) \wedge \bX)
\end{bmatrix}, \label{eqn:thresholdGT}
\end{alignat*}
where $\otimes$ and $\test(\cdot)$ are notations for the test operations in non-adaptive $(n, u)$-TGT; namely, $y_i := \cT(i, :) \otimes \bX := \test(\cT(i, :) \wedge \bX) = 1$ if $|\supp(\cT(1, :)) \cap \supp(\bX) \cap P| \geq u$ and $y_i := \cT(i, :) \otimes \bX := \test(\cT(i, :) \wedge \bX) = 0$ if $|\supp(\cT_{i, *}) \cap \supp(\bX) \cap P| < u$, for $i = 1, \ldots, t$. When $u = 1$, $(n, u)$-TGT reduces to standard group testing and we use $\odot$ instead of $\otimes$.

The procedure to get outcome vector $\bY$ is called \textit{encoding} and the procedure to recover $\bX$ from $\bY$ and $\cT$ is called \textit{decoding}. Our objective is to minimize $t$ and lower the time for exactly recovering $P = \supp(\bX)$ from $\bY$ by using $\cT$ with non-adaptive or 2-stage designs in $(n, d, u)$-TGT.

\subsection{Contributions}
\label{sub:intro:contri}

To the best of our knowledge, in the combinatorial setting and the deterministic design, this is the first work that solves the threshold group testing without gap problem with decoding complexity that scales in $u, d$, and $\log{n}$. When $u$ is a constant, we present a non-adaptive design in which the number of tests and the decoding time are $O(d^3 (\log^2{n}) \log{(n/d)} )$ and $O\big(d^3 (\log^2{n}) \log{(n/d)} + d^2 (\log{n}) \log^3{(n/d)} \big)$, respectively. Let $\min(g, h) = \min\{ \log{g}, (\log_h{g})^2 \}$ and $\min^*(u, d) = \min\left\{ u^u, (d - u)^{d - u} \right\}$. When $u$ is not a constant, the number of tests and the decoding time are $O \left( d(d - u)^2 u^2 \log{\frac{n}{d}} \times \min^*(u, d) \min(n, d - u + 1) \min(n, u) \right)$ and $O\big( d(d - u) u^2 \log{\frac{n}{d}} \min^*(u, d) $ $\min(n, u) \big( (d - u) \min(n, d - u + 1) + \log^2{\frac{n}{d - u}} \big) \big)$, respectively. These results are summarized in the following theorem.

\begin{theorem}
Consider $(n, d, u)$-TGT, i.e., TGT without gap and the number of defective items is $d$. There exists a non-adaptive design that can recover the defective set with $O \left( d(d - u)^2 \log{\frac{n}{d}} \min(n, d - u + 1) \min(n, u) \min^*(u, d) \right)$ tests in $O \big( d(d - u) u^2 \log{\frac{n}{d}} \min(n, u) \min^*(u, d) \big( (d - u) \min(n, d - u + 1) + \log^2{\frac{n}{d - u}} \big) \big)$ time, where $\min(g, h) = \min\{ \log{g}, (\log_h{g})^2 \}$ and $\min^*(u, d) = \min\left\{ u^u, (d - u)^{d - u} \right\}$.
\label{thm:withoutGap:nonadaptive}
\end{theorem}

For a special case when $u = 2$, with non-adaptive design, the number of tests and the decoding time are $O(d^3 (\log{n}) \log{(n/d)} )$ and $O(d^2 (\log{n} + \log^2{(n/d)}) )$, respectively. Moreover, with 2-stage design, the number of tests and the decoding time are $O(d^2 \log^2{(n/d)} )$. These results are presented in the theorem below.

\begin{theorem}
Consider $(n, d, \ell = 1, u = 2)$-TGT, i.e., TGT without gap and the number of defective items is $d$. There exists a non-adaptive design that can recover the defective set with $O\big( d^3 \log{\frac{n}{d}} \cdot \min(n, d) \big)$ tests in $O\big( d^2 \big(\log^2{\frac{n}{d}} + \min(n, d) \big) \big)$ time. For the 2-stage design, there exists a deterministic procedure that recovers the defective set with $O\big( d^2 \log^2{\frac{n}{d}} \big)$ tests in $O\big( d^2 \log^2{\frac{n}{d}} \big)$ time.
\label{thm:u2Det}
\end{theorem}

Consider TGT without gap and with $|P| \leq d$. For non-adaptive design, D'yachkov et al.~\cite{d2013superimposed} and Cheraghchi~\cite{cheraghchi2013improved} shows that it is possible to recover $P$ with $O( d^2 \log{n})$ tests when $u$ is constant. Chen et al.~\cite{chen2009nonadaptive} propose a design with $t = O \big( \big( \frac{d}{u} \big)^u \big( \frac{d}{d-u} \big)^{d - u} d \log{\frac{n}{d}} \big)$ tests to recover $P$ in $O(t u n^u)$ time, where $\mathrm{e}$ is the natural logarithm base. Although the number of tests is exponential to $d$, the decoding time is exponential to $d$ and linearly scales to $n^u$, this is the first work that tackle TGT under the combinatorial setting, the deterministic design, and the exact recovery criterion. Bui et al.~\cite{bui2019efficiently} later reduce the decoding complexity to $t \poly(d, \log{n})$. The numbers of tests in these works are much larger than in our proposed designs. A comparison of our proposed designs with previous works is summarized in Table~\ref{tbl:cmp}.

\begin{table*}[t]
\begin{center}
\scalebox{0.9}{
\begin{tabular}{|c|c|c|l|l|c|}
\hline
\begin{tabular}[c]{@{}c@{}} No. of\\ defectives \end{tabular} & \begin{tabular}[c]{@{}c@{}} No. of\\ stages \end{tabular} & \begin{tabular}[c]{@{}c@{}}Threshold\\ ($u$)\end{tabular} & \multicolumn{1}{c|}{Scheme} & \multicolumn{1}{c|}{\begin{tabular}[c]{@{}c@{}}No. of tests\\ ($t$)\end{tabular}} & \multicolumn{1}{c|}{\begin{tabular}[c]{@{}c@{}}Decoding complexity\end{tabular}} \\ \hline
\multirow{4}{*}{$d$} & \multirow{3}{*}{1} & \multirow{3}{*}{$\geq 2$} & De Marco et al.~\cite{de2020subquadratic} & $O \left( d^{3/2} \cdot \sqrt{\frac{d}{u}} \log{\frac{n}{d}} \right)$ & Not considered \\
\cline{4-6} & & & \textbf{Theorem~\ref{thm:withoutGap:nonadaptive}} & \begin{tabular}[c]{@{}c@{}} $O \left( d(d - u)^2 u^2 \log{\frac{n}{d}} \times \min(n, d - u + 1) \right.$ \\$\left. \times \min(n, u) \min\left\{ u^u, (d - u)^{d - u} \right\} \right)$ \end{tabular} & \begin{tabular}[c]{@{}c@{}} $O\big( d(d - u) u^2 \log{\frac{n}{d}} \min(n, u) \min\left\{ u^u, (d - u)^{d - u} \right\}$ \\$\times \big( (d - u) \min(n, d - u + 1) + \log^2{\frac{n}{d - u}} \big) \big)$ \end{tabular} \\
\cline{3-6} & & \multirow{2}{*}{2} & \textbf{Theorem~\ref{thm:u2Det}} & $O\big( d^3 \log{\frac{n}{d}} \cdot \min(n, d) \big)$ & $O\big( d^2 \big(\log^2{\frac{n}{d}} + \min(n, d) \big) \big)$ \\
\cline{2-2} \cline{4-6} & 2 & & \textbf{Theorem~\ref{thm:u2Det}} & $O\big( d^2 \log^2{\frac{n}{d}} \big)$ & $O\big( d^2 \log^2{\frac{n}{d}} \big)$ \\
\hline
\multirow{4}{*}{$\leq d$} & \multirow{4}{*}{1} & \multirow{4}{*}{$\geq 2$} & D'yachkov et al.~\cite{d2013superimposed} & $O \left( d^2 \log{n} \cdot \frac{(u-1)! 4^u}{(u - 2)^u (\ln{2})^u}\right)$ & Not considered \\
\cline{4-6} & & & Cheraghchi~\cite{cheraghchi2013improved} & $O((8u)^u d^2 \log{d} \cdot \log{\frac{n}{d}})$ & Not considered \\
\cline{4-6} & & & Chen et al.~\cite{chen2009nonadaptive} & $O \left( \left( \frac{d}{u} \right)^u \left( \frac{d}{d-u} \right)^{d - u} d \log{\frac{n}{d}} \right)$ & $O(t u n^u)$ \\
\cline{4-6} & & & Bui et al.~\cite{bui2019efficiently} & $O \left(  \left( \frac{d}{u} \right)^u \left( \frac{d}{d -u} \right)^{d - u} d^3 \log{n} \cdot \log{\frac{n}{d}} \right)$ & $t \times \poly(d, \log{n})$ \\
\hline
\end{tabular}}
\end{center}

\caption{Comparison with existing work under the deterministic design, the combinatorial setting, and the exact recovery setting. Here, we denote $\min(g, h) = \min\{ \log{g}, (\log_h{g})^2 \}$.}
\label{tbl:cmp}
\end{table*}

\subsection{Related work}
\label{sec:intro:related}

Consider a population of $n$ items and a set of up to $d$ defective items, i.e., $|P| \leq d$. In SGT, it is possible to obtain the information-theoretic number of tests $O(d\log{(n/d)})$ and identify all defective items by using 2-stage designs~\cite{de2005optimal}. In non-adaptive design, although the lower bound of the number of tests is $O(d^2 \log_d{n})$~\cite{d1982bounds}, the best practice for attaining explicit construction is $O(d^2 \log{n})$~\cite{porat2011explicit}. There is a body of literatures working on attaining $O(d^2 \log{n})$ tests and decoding time of $\poly(d, \log{n})$ such as~\cite{indyk2010efficiently,ngo2011efficiently,cheraghchi2020combinatorial,cai2017efficient,bondorf2020sublinear}.

Consider TGT without gap. When $|P| = d$, De Marco et al.~\cite{de2020subquadratic} propose a scheme with no decoding procedure that requires $O \left( d^{3/2} \cdot \sqrt{\frac{d}{u}} \log{\frac{n}{d}} \right)$ tests. There are some works~\cite{chan2013stochastic,reisizadeh2018sub} that consider the case $|P| = d$ under the probabilistic setting with some models applied on the gap $g$ and achieve lower numbers of tests compared to ours. However, either their decoding complexity scales linearly to the number of items~\cite{chan2013stochastic} or an external look-up matrix of size $O(u \log{n}) \times \binom{n}{u}$ is required~\cite{reisizadeh2018sub}.

\section{Preliminaries}
\label{sec:pre}

\subsection{Notations}
\label{sub:pre:notation}

For consistency, we use capital calligraphic letters for matrices, non-capital letters for scalars, bold letters for vectors, and capital letters for sets. All matrix and vector entries are binary. The frequently used notations are listed in Table~\ref{buice.t2}.

\begin{table}[t]
\begin{center}
\scalebox{1.0}{
\begin{tabular}{|c|l|}
\hline
\begin{tabular}{@{}c@{}} Notation \end{tabular} & \begin{tabular}{@{}c@{}} Description \end{tabular} \\
\hline
$n$ & Number of items \\
\hline
$d$ & Number of defective items \\
\hline
$\bX = (x_1, \ldots, x_n)^T$ & Binary representation of $n$ items \\
\hline
$\supp(\bV) = \{j \mid v_j \neq 0 \}$ & The support set for vector $\bV = (v_1, \ldots, v_w)$ \\
\hline
$\vecc(A)$ & A $1 \times n$ support vector for set $A$ such that $\supp(\vecc(A)) = A$. \\
\hline
$u$ & Threshold in non-adaptive $(n, u)$-TGT model  \\
\hline
$P = \{p_1, p_2, \ldots, p_{|P|} \}$ & \begin{tabular}{@{}l@{}} Set of defective items; Cardinality of $P$ is $|P| = d$ \end{tabular}  \\
\hline
$N = [n] = \{1, \ldots, n \}$ & Set of $n$ items \\
\hline
$\otimes_u$ & \begin{tabular}{@{}l@{}} Operation related to non-adaptive $(n, u)$-TGT (to be defined later) \end{tabular}  \\
\hline
$\cT(i, :), \cM(i, :)$ & Row $i$ of matrix $\cT$, row $i$ of matrix $\cM$ \\
\hline
$\cT(:, j), \cM(:, j)$ & Column $j$ of matrix $\cT$, column $j$ of matrix $\cM$ \\
\hline
$\min(g, h)$ & $\min\{ \log{g}, (\log_h{g})^2 \}$ \\
\hline
$\mathrm{e}$ & The natural logarithm base. \\
\hline
\end{tabular}}

\end{center}
\caption{Notations frequently used in this paper.}

\label{buice.t2}
\end{table}

\subsection{Disjunct and list-disjunct matrices}
\label{sub:pre:disjunct}

Disjunct matrices are a powerful tool to tackle the standard and threshold group testing problem~\cite{chen2009nonadaptive,cheraghchi2013improved,bui2019efficiently}. They were first introduced by Kautz and Singleton~\cite{kautz1964nonrandom} and then generalized by Stinson and Wei~\cite{stinson2004generalized} and D'yachkov et al.~\cite{d2002families}. The formal definition of a disjunct matrix is as follows.

\begin{definition}
An $m \times n$ binary matrix $\cM$ is called a $k$-disjunct matrix if, for any two disjoint subsets $S_1, S_2 \subset [n]$ such that $|S_1| = k$ and $|S_2| = 1$, there exists at least one row in which there are all 1's among the columns in $S_2$ while all the columns in $S_1$ have 0's, i.e., $\left\vert \bigcap_{j \in S_2} \supp \left( \cM(:, j) \right) \big\backslash \bigcup_{j \in S_1} \supp \left( \cM(:, j) \right) \right\vert \geq 1$.
\label{def:threshDisjunct}
\end{definition}

When $\cM$ is a $d$-disjunct matrix and $|\bX| \leq d$, there exists an efficient procedure to recover $\supp(\bX)$ from $\cM \odot \bX$ as follows.

\begin{theorem}~\cite[Theorem 3.3]{cheraghchi2020combinatorial}
There exists a Monte-Carlo construction of a $k$-disjunct matrix $\cM \in \{0, 1 \}^{m \times n}$ with $m = O(k^2 \min\{\log{n}, (\log_k{n})^2 \})$, that allows decoding in $O(m + k\log^2{(n/k)} )$ time.
\label{thm:disjunct}
\end{theorem}

List-disjunct matrices are a relaxation of disjunct matrices~\cite{dyachkov1983survey,rashad1990random,de2005optimal,indyk2010efficiently,cheraghchi2013noise}. In particular, once we set $|S_2| = \ell \geq 1$ in Definition~\ref{def:threshDisjunct}, matrix $\cM$ is called a $(k, \ell)$-list disjunct matrix. Given an input vector $\bX$ with $P = \supp(\bX)$ and the test outcome vector $\cM \odot \bX$, there exists a decoding procedure that returns set $L$ such that $L \cap P = P$ and $|L| \leq k + \ell - 1$. Cheraghchi and Nakos~\cite{cheraghchi2020combinatorial} present an efficient encoding and decoding scheme for list disjunct matrices as follows.

\begin{theorem}~\cite[Theorem 3.2]{cheraghchi2020combinatorial}
There exists a Monte-Carlo construction of a $(k, k)$-list disjunct matrix $\cM \in \{0, 1 \}^{m \times n}$ with $m = O(k \log{(n/k)})$, that allows decoding in $O(k\log^2{(n/k)} )$ time.
\label{thm:listDisjunct}
\end{theorem}

For any vector $\bX$ with $|\bX| \leq k$, let $\mathrm{Dec}(\cM \odot \bX, \cM)$ be an exact recovery decoding procedure that returns $\supp(\bX)$ when $\cM$ is a $k$-disjunct matrix or returns set $L$ with $L \cap \supp(\bX) = \supp(\bX)$ and $|L| \leq 2k$ when $\cM$ is a $(k, k)$-list disjunct matrix.

\section{Non-adaptive design}
\label{sec:withoutGap}

In this section, we present an algorithm with up to $O(u^u d^3 \log^3{n})$ decoding time to find the defective set.

\subsection{Comments on the technique using auxiliary subset}
\label{sec:withoutGap:comments}

In the seminal work of threshold group testing, Damaschke~\cite{damaschke2006threshold} proposed an idea of using a set $U$ of $u - 1$ defective items as auxiliary pool in the standard group testing strategy. Then the number of tests and the decoding time for identifying the defective set are the same as the standard group testing strategy. This cannot be applied to non-adaptive setting because the defective items in the set of $u - 1$ defective items $U$ can not be identified by this way. Indeed, let $x \in U$ and $\cM$ be a $d$-disjunct matrix. Because of the disjunctiveness property of $\cM$, there exists a test in which it contains $x$ and none of $P \setminus \{ x \}$. By adding the $u - 1$ items in $U$ into this test, the number of defective items in this test is $u - 1$ because $x \in U$. Consequently, the outcome of this test is negative. Thus, the defective item $x$ is identified as negative. We reiterate that any work using the auxiliary subset proposed by Damaschke for the non-adaptive setting may encounter the same mistake. There is a simple way to remedy this. Given a known $u$ defective items, we use $u$ distinct auxiliary pools of $u - 1$ defective items in the standard group testing strategy. Then the defective set is the union set of $u$ defective sets obtained from the decoding procedures in the standard group testing strategy.

\subsection{Single selectors}
\label{sec:withoutGap:singleSelector}

To attain a design with an efficient encoding and decoding algorithm, we propose a notation called \emph{single selector}.

\begin{definition}
Given integers $m, d$, and $n$, with $1 \leq m \leq d \leq n$. A Boolean matrix is called \emph{an $(m, d, n)$-single selector} if for any arbitrary set $D \subseteq [n]$ with $|D| = d$, there exists a row in $\cM$ such that it contains exactly $m$ items in $D$.
\label{def:singleSelector}
\end{definition}

The following lemma presents the existence of single selectors by probabilistic construction (proved in the appendix).

\begin{lemma}
Given integers $m, d$, and $n$, with $1 \leq m < d \leq n$, there exists an $(m, d, n)$-single selector with $O\big( \min\left\{ m^m, (d - m)^{d - m} \right\} d \log{\frac{n}{d}} \big)$ rows. When $m = d$, $t = O\big( d \log{\frac{n}{d}} \big)$.
\label{lem:singleSelection}
\end{lemma}

\begin{proof}

Consider a binary Bernoulli matrix $\cM = (m_{ij})$ of size $t \times n$ in which $\Pr[m_{ij} = 1] = p$ and $\Pr[m_{ij} = 0] = 1 - p$. Consider a set $D$ of $d$ columns in $\cM$. Let the event that a row does not contain exactly $m$ ones in $D$ be \emph{the bad event}. The probability that the bad event happens is:
\[
1 - \binom{d}{m} p^m (1 - p)^{d - m}.
\]

By setting $p = 1/d$, we get
\begin{align}
&\binom{d}{m} p^m (1 - p)^{d - m} = \binom{d}{m} \left( \frac{1}{d} \right)^{m} \left( 1 - \frac{1}{d} \right)^{d - m} \nonumber \\
&\geq \left( \frac{d}{m} \right)^m \left( \frac{1}{d} \right)^{m} \left( 1 - \frac{1}{d} \right)^{d - m} = \frac{1}{m^m} \left( 1 - \frac{1}{d} \right)^d \left( 1 + \frac{1}{d - 1} \right)^m \label{eqn:binom} \\
&\geq \frac{1}{m^m} \times \mathrm{e} \left( 1 - \frac{1}{d} \right) \left( 1 + \frac{1}{d - 1} \right)^m = \frac{\mathrm{e}}{m^m} \left( 1 + \frac{1}{d - 1} \right)^{m - 1} \label{eqn:lowerExp} \\
&\geq \frac{\mathrm{e}}{m^m}, \nonumber
\end{align}
where $\mathrm{e}$ is base of the natural logarithm. Equation~\eqref{eqn:binom} is obtained because $\left( \frac{a}{b} \right)^b \leq \binom{a}{b}$ for $0 < a \leq b$. We get~\eqref{eqn:lowerExp} because $\mathrm{e} \big(1 - \frac{1}{d} \big) \leq \big(1 - \frac{1}{d} \big)^d$. Consequently, the probability that the bad event happens for the $t$ rows with respect to $D$ is
\begin{equation}
\left( 1 - \binom{d}{m} p^m (1 - p)^{d - m} \right)^t \leq \left( 1 - \frac{\mathrm{e}}{m^m} \right)^t \leq \exp\left( - \frac{\mathrm{e} t}{m^m} \right), \nonumber
\end{equation}
because $1 - x \leq e^{-x}$ for any $x > 0$. Therefore, the probability that $\cM$ is not an $(m, d, n)$-single selector, i.e., there exists at least one subset $D \subseteq [n]$ with $|D| = d$ such that the bad event happens for the $t$ rows, is up to:
\begin{equation}
\binom{n}{d} \left( 1 - \binom{d}{m} p^m (1 - p)^{d - m} \right)^t \leq \left( \frac{\mathrm{e} n}{d} \right)^d \exp\left( - \frac{\mathrm{e} t}{m^m} \right), \label{eqn:dt}
\end{equation}
where the last equation is obtained because $\binom{a}{b} \leq \left( \frac{\mathrm{e} a}{b} \right)^b$ for $0 < a \leq b$. 
To ensure that there exists an $(m, d, n)$-single selector $\cM$, one needs
\begin{equation}
\left( \frac{\mathrm{e} n}{d} \right)^d \exp\left( - \frac{\mathrm{e} t}{m^m} \right) < 1 \Longleftrightarrow \frac{m^m d}{\mathrm{e}} \left(1 + \log{\frac{n}{d}} \right) < t. \label{eqn:tt}
\end{equation}

Thus, we can choose $t = \lceil \frac{m^m d}{\mathrm{e}} \big(1 + \log{\frac{n}{d}} \big) \rceil + 1 = O\big( m^m d \log{\frac{n}{d}} \big)$ to satisfy~\eqref{eqn:tt}.

We can easily create a $(d - m, d, n)$-single selector $\bar{\cM} = (\bar{m}_{ij})$ from a $t \times n$ $(m, d, n)$-single selector $\cM = (m_{ij})$ by inverting all entries in $\cM$, i.e., $\bar{m}_{ij} = 1 - m_{ij}$ for $i \in [t]$ and $j \in [n]$.

When $m = d$, we choose $p = d/(d+1)$ then~\eqref{eqn:dt} becomes 
\begin{align}
&\binom{n}{d} \left( 1 - \binom{d}{m} p^m (1 - p)^{d - m} \right)^t \nonumber \\
&= \binom{n}{d} \left( 1 - \left( \frac{d}{d + 1} \right)^d \right)^t \leq \left( \frac{\mathrm{e} n}{d} \right)^d \exp\left( - t \left( \frac{d}{d + 1} \right)^d \right).\nonumber 
\end{align}
This inequality holds if we choose $t = \lceil \mathrm{e} d \left( 1 + \log{\frac{n}{d}} \right) \rceil + 1 = O\big( d \log{\frac{n}{d}} \big)$ because:
\begin{align}
&t > \mathrm{e} d \left( 1 + \log{\frac{n}{d}} \right) \geq \left(1 + \frac{1}{d} \right)^d \log{\left( \frac{\mathrm{e} n}{d} \right)^d} \nonumber \\
&\Longrightarrow \left( \frac{\mathrm{e} n}{d} \right)^d \exp\left( - \frac{\mathrm{e} t}{m^m} \right) < 1. \nonumber
\end{align}
This completes our proof.

\end{proof}

\subsection{Non-adaptive design}
\label{sec:withoutGap:nonadaptive}

Here we present a non-adaptive design to tackle TGT without gap and summarize it in Theorem~\ref{thm:withoutGap:nonadaptive}. When $u = d$, it has been known that the $d$ defective items can be found by using a $d$-disjunct matrix and its complementary matrix~\cite[Lemma 1]{bui2019efficiently}. Because of Theorem~\ref{thm:disjunct}, it takes $O(d^2 \min(d, n))$ tests to identify all $d$ defective items in $O(d^2 \min(d, n) + d\log^2{(n/d)} )$ time. This satisfies the claims in Theorem~\ref{thm:withoutGap:nonadaptive}. Therefore, we only consider the case $u < d$ to prove Theorem~\ref{thm:withoutGap:nonadaptive}.

\subsubsection{Main idea}
\label{sec:withoutGap:nonadaptive:idea}

Our objective is first creating a test that contains exactly $u$ defective items. From these $u$ defective items, we create $u$ sets $S_1, \ldots, S_u \subseteq [n]$ of items such that each set contains exactly $u - 1$ defective items and the sets of the defective items in these $u$ tests are distinct, i.e., $S_i \cap P \neq S_{i^\prime} \cap P$ for any $i \neq i^\prime$. The by augmenting $S_i$ to a $(d - u)$-disjunct matrix $\cM$ for testing and decoding, we can recover set $P \setminus S_i$. Therefore, the defective set $P$ can be recover by setting $P = \bigvee_{i = 1}^u (P \setminus S_i)$.

\subsubsection{Encoding} In the pursuit of attaining an efficient decoding procedure, we initialize the following matrices:
\begin{itemize}
\item Matrix $\cA$ of size $a \times n$ is a $(d, u, n)$-single selector, where $a = O\left( \min\left\{ u^u, (d - u)^{d - u} \right\}d \log{\frac{n}{d}} \right)$ and set $0^0 = 1$ for simplicity.
\item Matrix $\cM$ and $\cB$ of size $m \times n$ and $b \times n$ are $(d - u + 1)$-disjunct and $u$-disjunct matrices, respectively. Specifically, $m = O \big( (d - u)^2 \min(n, d - u + 1) \big)$ and $b = O \big( u^2 \min(n, u) \big)$.
\end{itemize}

The whole encoding procedure is illustrated in Fig.~\ref{fig:Enc}. In this figure, we set $u = 3$, $d = 5$, and $P = \{ j_1, j_2, j_3, j_4, j_5 \}$. Matrix $\cA$ has some special row that contains only $u$ defective items, says $j_1, j_3$, and $j_5$, in the defective set $P$. Let $S = \{j_1, j_2, j_5 \}$. Since matrix $\cB$ is $u$-disjunct, there exists $u = 3$ distinct private rows such that the first row (squared by a solid rectangle) contains $j_1$ but does not contain the items in $S \setminus \{ j_1 \}$, the second row (squared by a long dash dot dot rectangle) contains $j_3$ but does not contain the items in $S \setminus \{ j_3 \}$, and the third row (squared by a square dot rectangle) contains $j_5$ but does not contain the items in $S \setminus \{ j_5 \}$. By applying the exclusion operator (defined later) to $\cA$ and $\cB$, the special row in $\cA$ and the $u$ private rows in $\cB$ generate $u$ rows in which each row contains $u - 1$ items in $S$. The set of the defective items in each row is the complementary set of the defective set in its corresponding private row. In particular, the first, second, and third rows contain $S \setminus \{ j_1 \}, S \setminus \{ j_3 \}$, and $S \setminus \{ j_5 \}$, respectively. These rows are depicted in $u=3$ matrices and their corresponding solid, long dash dot dot, and square dot rectangles, after the first three arrows. By augmenting each row with a $(d - u + 1)$-disjunct matrix $\cM$, the new matrix will contain at least $d - u + 1$ rows, each row contains $u$ defective items. Among these $u$ defective items, there always exist $u - 1$ defective items in the row augmented with $\cM$. Testing each new matrix result in an outcome vector that can be converted into tests in standard group testing. From example, for the row containing $S \setminus \{ j_1 \}$, the outcome vector for testing the new matrix generated by it and $\cM$ is $\cM = \vecc(P \setminus (S \setminus \{ j_1 \}) )$.

\begin{figure*}
\centering
\includegraphics[scale=0.5]{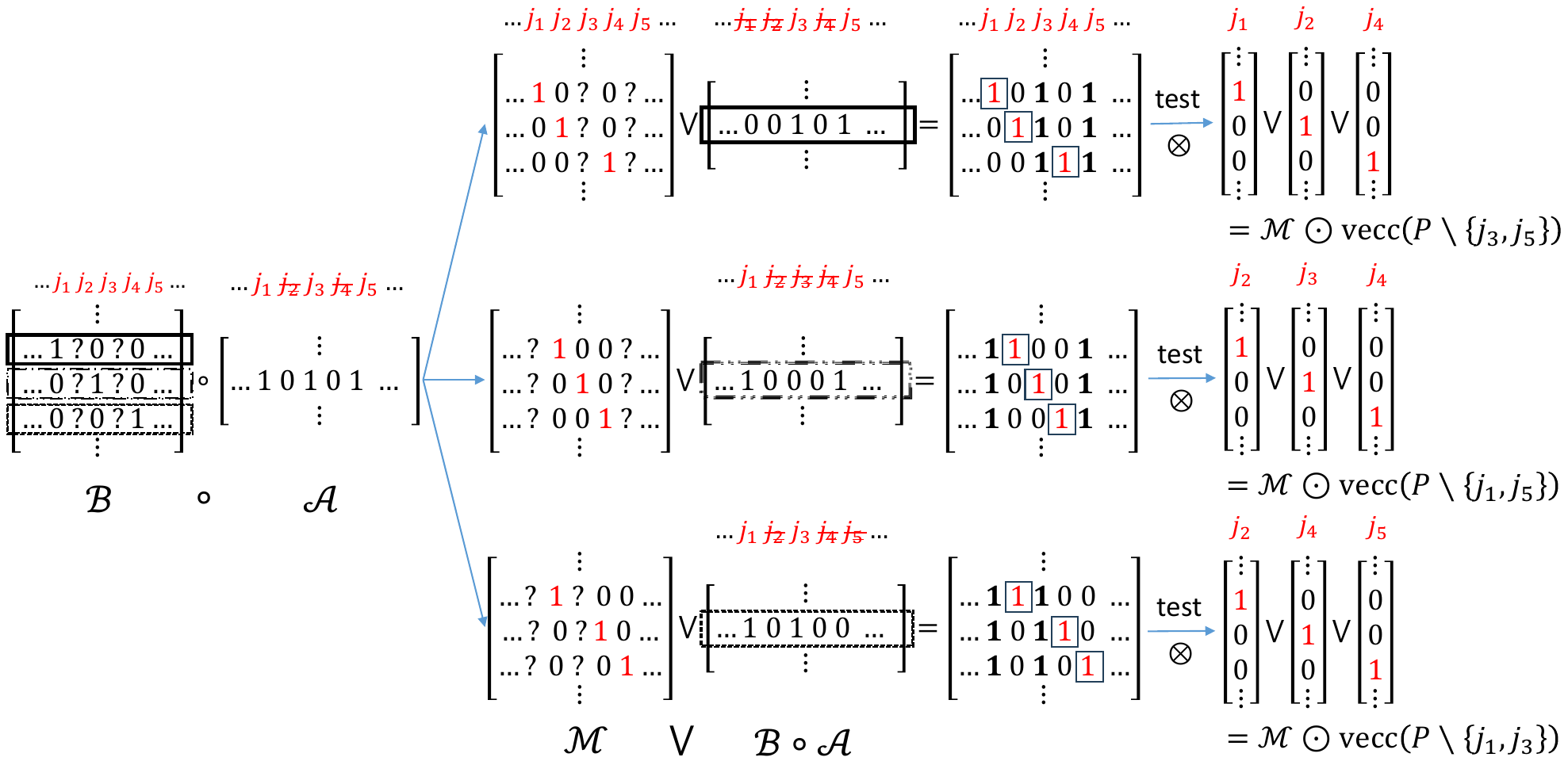}
\caption{Illustration of the encoding procedure for $u = 3, d = 5$, and $P = \{ j_1, j_2, j_3, j_4, j_5 \}$.}
\label{fig:Enc}
\end{figure*}

We define some notations as follows. The exclusion-wise operator between two vectors of same size $\bZ = (z_1, \ldots, z_n)$ and $\bZ^\prime = (z_1^\prime, \ldots, z_n^\prime)$ is $\bZ \circ \bZ^\prime = (z_1^\prime \oplus (z_1^\prime \wedge z_1), \ldots, z_n^\prime \oplus (z_n^\prime \wedge z_n))$, where $\oplus$ is the XOR-bitwise operator. In other words, $\supp(\bZ \circ \bZ^\prime) = \supp(\bZ^\prime) \setminus \supp(\bZ)$. Then the exclusion operator between two matrices $\cA$ and $\cB$ is defined  as $\cB \circ \cA := \begin{bmatrix}
\cB \circ \cA(1, :) \\
\vdots \\
\cB \circ \cA(a, :)
\end{bmatrix}$, where $\cB \circ \cA(i, :) := \begin{bmatrix}
\cB(1, :) \circ \cA(i, :) \\
\vdots \\
\cB(b, :) \circ \cA(i, :) \end{bmatrix}$ for $i = 1, \ldots, a$.

Similarly, the OR-wise operator between two vectors of same size $\bZ = (z_1, \ldots, z_n)$ and $\bZ^\prime = (z_1^\prime, \ldots, z_n^\prime)$ is $\bZ \vee \bZ^\prime = (z_1 \vee z_1^\prime, \ldots, z_n \vee z_n^\prime)$. In other words, $\supp(\bZ \vee \bZ^\prime) = \supp(\bZ) \cup \supp(\bZ^\prime)$. The OR operator between two matrices $\cA$ and $\cB$ is then defined as $
\cB \vee \cA := \begin{bmatrix}
\cB \vee \cA(1, :) \\
\vdots \\
\cB \vee \cA(a, :)
\end{bmatrix}$, where $\cB \vee \cA(i, :) := \begin{bmatrix}
\cB(1, :) \vee \cA(i, :) \\
\vdots \\
\cB(b, :) \vee \cA(i, :) \end{bmatrix}$ for $i = 1, \ldots, a$.

The measurement matrix is defined as
\begin{equation}
\cT :=
\begin{bmatrix}
\cA \\
\cB \circ \cA \\
\cM \vee (\cB \circ \cA)
\end{bmatrix}, \nonumber 
\end{equation}
and the outcome vector is
\begin{equation}
\bY := \test(\cT) =
\begin{bmatrix}
\vdots \\
\test(\cA(i, :)) \\
\test(\cB(i^\prime, :) \circ \cA(i, :)) \\
\test(\cM \vee (\cB(i^\prime, :) \circ \cA(i, :))) \\
\vdots
\end{bmatrix}, \label{outcomeVector}
\end{equation}
where $i = 1, \ldots, a$ and $i^\prime = 1, \ldots, b$.

\subsubsection{Decoding} For each $i \in \{1, \ldots, a \}$ and $i^\prime \in \{1, \ldots, b \}$, if $\test(\cA(i, :))$ is positive and $\test(\cB(i^\prime, :) \circ \cA(i, :))$ is negative, set $S_{i, i^\prime} = \dec(\test(\cM \vee (\cB(i^\prime, :) \circ \cA(i, :))), \cM)$ and add $S_{i, i^\prime}$ into the defective set $S$.

\subsubsection{Correctness}

We first state a useful lemma in utilizing disjunct matrices to identify defective items with auxiliary set. The full proof and complexity are presented in the full version.

\begin{lemma}
Consider $(n, \bar{p}, u)$-TGT, i.e., TGT without gap and the number of defective items is at most $p$. Let $P$ and $\cM$ be a defective set and a $t \times n$ $p$-disjunct matrix, where $|P| \leq p$. Let $\bX \in \{0, 1 \}^n$ be an input vector and $\bY = \test(\cM \vee \bX)$. Then if $\test(\bX) = 0$, for any decoding procedure with the exact recovery in standard group testing with the measurement matrix $\cM$ and the outcome vector $\bY$, the returned defective set is $P \setminus \supp(\bX)$ if $|\supp(\bX) \cap P| = u - 1$ and an empty set if $|\supp(\bX) \cap P| < u - 1$.
\label{lem:CGTwithSet}
\end{lemma}

As described in Section~\ref{sec:withoutGap:nonadaptive:idea}, we need to prove that there exist $u$ sets $S_1, \ldots, S_u \subseteq [n]$ such that each set contains exactly $u - 1$ defective items and $S_i \cap P \neq S_{i^\prime} \cap P$ for any $i \neq i^\prime$. Indeed, since $\cA$ is a $(d, u, n)$-single selector and $|P| = d$, there exists a row $r^*$ such that $|\supp(\cA(r^*, :)) \cap P| = u$. Let $Q = \supp(\cA(r^*, :)) \cap P$. Because $\cB$ is a $u$-disjunct matrix, for any $q \in Q$, there exists a row, denoted as $r_q$, such that $\supp(\cB(r_q, :)) \cap Q = \{ q \}$. On the other hand, $\supp(\cB(r_q, :) \circ \cA(r^*, :)) = (\supp(\cA(r^*, :)) \setminus \supp(\cB(r_q, :)) \cup (\supp(\cB(r_q, :) \setminus \supp(\cA(r^*, :)) ) $. Therefore, $\supp(\cB(r_q, :) \circ \cA(r^*, :)) \cap P = (\supp(\cA(r^*, :)) \cap P) \setminus (\supp(\cB(r_q, :) \cap P) = Q \setminus \{ q \}$. Since $|Q| = u$, there thus exist $u$ sets $S_1, \ldots, S_u \subseteq [n]$ such that each set contains exactly $u - 1$ defective items and $S_i \cap P \neq S_{i^\prime} \cap P$ for any $i \neq i^\prime$. On the other hand, $\test(\cA(r^*, :))$ is positive and $\test(\cB(r_q, :) \circ \cA(r^*, :))$ is negative, by using Lemma~\ref{lem:CGTwithSet} with the fact that $\cM$ is a $(d - u + 1)$-disjunct matrix, we get $\test(\cM \vee (\cB(r_q, :) \circ \cA(r^*, :))) = \test(\cM, Q \setminus \{ q \})$. Therefore, set $S_{r^*, r_q} = \dec(\test(\cM \vee (\cB(r_q, :) \circ \cA(r^*, :))), \cM)$ is $P \setminus (Q \setminus \{ q \})$. It is straightforward that $P = \bigcup_{q \in Q} P \setminus (Q \setminus \{ q \})$. This implies the set $S$ always contains all defective items.

To complete our proof, we need to show that the decoding procedure never adds false defective items to the defective set. Consider $i \in [a]$. If $|\supp(\cA(i, :)) \cap P| < u$ then $\test(\cA(i, :))$ is negative. Therefore, no items are added to $S$. If $\test(\cA(i, :))$ is positive and $\test(\cB(i^\prime, :) \circ \cA(i, :))$ is positive, again, no items are added to $S$. If $\test(\cA(i, :))$ is positive and $\test(\cB(i^\prime, :) \circ \cA(i, :))$ is negative, there are only two cases:
\begin{itemize}
\item $|\cA(i, :) \cap P| = u$ and $|\supp(\cB(i^\prime, :) \circ \cA(i, :)) \cap P| < u$.
\item $|\cA(i, :) \cap P| > u$ and $|\supp(\cB(i^\prime, :) \circ \cA(i, :)) \cap P| < u$.
\end{itemize}

In both cases, by Lemma~\ref{lem:CGTwithSet}, the returned defective set by using any decoding procedure with the exact recovery in standard group testing never contains negative items.

\subsubsection{Complexity}

By using Lemma~\ref{lem:singleSelection}, $a = O\left( \min\left\{ u^u, (d - u)^{d - u} \right\}d \log{\frac{n}{d}} \right)$, where we set $0^0 = 1$ for simplicity. As in Theorem~\ref{thm:disjunct}, a $t \times n$ $p$-disjunct matrix $\cM$ can be efficiently decoded in $O \big( t + p \log^2{(n/p)} \big)$ time with $t = O \big( p^2 \min(n, p) \big)$ where $\min(g, h) = \min\{ \log{g}, (\log_h{g})^2 \}$. Therefore, we get $m = O \big( (d - u)^2 \min(n, d - u + 1) \big)$ and $b = O \big( u^2 \min(n, u) \big)$. As a result, the number of tests is:
\begin{align}
&a + b \times a + m \times b \times a \\
&= O \left( d(d - u)^2 u^2 \log{\frac{n}{d}} \times \min(n, d - u + 1) \min(n, u) \times \min\left\{ u^u, (d - u)^{d - u} \right\} \right) \nonumber \\
&= O \left( d^3 (\log^2{n}) \log{\frac{n}{d}} \right), \mbox{ if $u$ is a constant.}
\end{align}

The decoding complexity is
\begin{align}
&O(a) + O(b \times a) + O(b \times a ) \times O \big( (d - u)^2 \min(n, d - u + 1) + (d - u) \log^2{\frac{n}{d - u}} \big) \nonumber \\
&= O \left( d(d - u) u^2 \log{\frac{n}{d}} \min(n, u) \min\left\{ u^u, (d - u)^{d - u} \right\} \times \left( (d - u) \min(n, d - u + 1) + \log^2{\frac{n}{d - u}} \right) \right). \nonumber
\end{align}

\subsubsection{Example} We delineate our encoding and decoding procedures by an example here.

\emph{Encoding:} Suppose $n = 12, d = 3, u = 2$, and the defective set is $P = \{1, 8, 11 \}$. Matrix $\cA$ of size $6 \times 12$ is a $(d = 3, u = 2, n = 12)$-single selector and matrices $\cB$ and $\cM$ of size $9 \times 12$ are $(d + u - 1) = u = 2$-disjunct matrices. We get $a = 6, m = b = 9$. Because the number of rows in $\cT$ is $amb = 486$ which is large to be written in this paper, we only demonstrate some columns and rows of interest here. Matrices $\cA, \cB$, and $\cM$ are given as follows.

\begin{align}
\cA &= \left[ \begin{tabular}{cccccccccccc}
\textbf{0} & 1 & 0 & 1 & 0 & 0 & 1 & \textbf{0} & 0 & 0 & \textbf{0} & 1 \\
\textbf{1} & 0 & 0 & 1 & 0 & 1 & 0 & \textbf{0} & 1 & 1 & \textbf{0} & 0 \\
\textbf{1} & 0 & 1 & 0 & 1 & 1 & 0 & \textbf{1} & 1 & 1 & \textbf{1} & 0 \\
\textbf{0} & 1 & 1 & 0 & 1 & 0 & 1 & \textbf{1} & 0 & 0 & \textbf{1} & 1 \\
\textbf{1} & 1 & 0 & 0 & 1 & 1 & 0 & \textbf{0} & 0 & 0 & \textbf{1} & 1 \\
\textbf{0} & 0 & 1 & 1 & 0 & 0 & 1 & \textbf{1} & 1 & 1 & \textbf{0} & 0
\end{tabular}
\right], \nonumber \\
\cB = \cM &= \left[
\begin{tabular}{cccccccccccc}
\textbf{0} & 0 & 0 & 0 & 0 & 0 & 1 & \textbf{1} & 1 & 1 & \textbf{0} & 0 \\
\textbf{0} & 0 & 0 & 1 & 1 & 1 & 0 & \textbf{0} & 0 & 1 & \textbf{0} & 0 \\
\textbf{1} & 1 & 1 & 0 & 0 & 0 & 0 & \textbf{0} & 0 & 1 & \textbf{0} & 0 \\
\textbf{0} & 0 & 1 & 0 & 0 & 1 & 0 & \textbf{0} & 1 & 0 & \textbf{1} & 0 \\
\textbf{0} & 1 & 0 & 0 & 1 & 0 & 0 & \textbf{1} & 0 & 0 & \textbf{1} & 0 \\
\textbf{1} & 0 & 0 & 1 & 0 & 0 & 1 & \textbf{0} & 0 & 0 & \textbf{1} & 0 \\
\textbf{0} & 1 & 0 & 1 & 0 & 0 & 0 & \textbf{0} & 1 & 0 & \textbf{0} & 1 \\
\textbf{0} & 0 & 1 & 0 & 1 & 0 & 1 & \textbf{0} & 0 & 0 & \textbf{0} & 1 \\
\textbf{1} & 0 & 0 & 0 & 0 & 1 & 0 & \textbf{1} & 0 & 0 & \textbf{0} & 1 
\end{tabular}
\right]. \nonumber
\end{align}

Since $P = \{1, 8, 11 \}$, we only consider columns $1, 8$, and $11$ in any matrix considered later. Let us denote $\cY \mid_{P}$ be the pruning matrix of $\cY$ restricted by the columns whose indices are in $P$. The output matrices between $\cB$ and $\cA$ restricted to $P$ by using the exclusion operator are:

\begin{alignat}{3}
(\cB \circ \cA(1, :))\mid_P &= \left[ \begin{tabular}{ccc}
0 & 0 & 0 \\
0 & 0 & 0 \\
0 & 0 & 0 \\
0 & 0 & 0 \\
0 & 0 & 0 \\
0 & 0 & 0 \\
0 & 0 & 0 \\
0 & 0 & 0 \\
0 & 0 & 0
\end{tabular}
\right],
(\cB \circ \cA(2, :))\mid_P &= \left[ \begin{tabular}{ccc}
1 & 0 & 0 \\
1 & 0 & 0 \\
0 & 0 & 0 \\
1 & 0 & 0 \\
1 & 0 & 0 \\
0 & 0 & 0 \\
1 & 0 & 0 \\
1 & 0 & 0 \\
0 & 0 & 0
\end{tabular}
\right], (\cB \circ \cA(3, :))\mid_P &= \left[ \begin{tabular}{ccc}
1 & 0 & 1 \\
1 & 1 & 1 \\
0 & 1 & 1 \\
1 & 1 & 0 \\
\textbf{1} & \textbf{0} & \textbf{0} \\
\textbf{0} & \textbf{1} & \textbf{0} \\
1 & 1 & 1 \\
1 & 1 & 1 \\
\textbf{0} & \textbf{0} & \textbf{1}
\end{tabular}
\right], \nonumber \\
(\cB \circ \cA(4, :))\mid_P &= \left[ \begin{tabular}{ccc}
\textbf{0} & \textbf{0} & \textbf{1} \\
0 & 1 & 1 \\
0 & 1 & 1 \\
\textbf{0} & \textbf{1} & \textbf{0} \\
\textbf{0} & \textbf{0} & \textbf{0} \\
\textbf{0} & \textbf{1} & \textbf{0} \\
0 & 1 & 1 \\
0 & 1 & 1 \\
\textbf{0} & \textbf{0} & \textbf{1}
\end{tabular}
\right], (\cB \circ \cA(5, :))\mid_P &= \left[ \begin{tabular}{ccc}
1 & 0 & 1 \\
1 & 0 & 1 \\
\textbf{0} & \textbf{0} & \textbf{1} \\
\textbf{1} & \textbf{0} & \textbf{0} \\
\textbf{1} & \textbf{0} & \textbf{0} \\
\textbf{0} & \textbf{0} & \textbf{0} \\
1 & 0 & 1 \\
1 & 0 & 1 \\
\textbf{0} & \textbf{0} & \textbf{1}
\end{tabular}
\right], (\cB \circ \cA(6, :))\mid_P &= \left[ \begin{tabular}{ccc}
0 & 0 & 0 \\
0 & 1 & 0 \\
0 & 1 & 0 \\
0 & 1 & 0 \\
0 & 0 & 0 \\
0 & 1 & 0 \\
0 & 1 & 0 \\
0 & 1 & 0 \\
0 & 0 & 0
\end{tabular}
\right] \nonumber
\end{alignat}

For $u = 2$, the test outcomes for matrix $\cA$ is $\test(\cA(1, :)) = \test(\cA(2, :)) = \test(\cA(6, :))= 0, \test(\cA(3, :)) = \test(\cA(4, :)) = \test(\cA(5, :)) = 1$. For the tests on $\cM \vee (\cB \circ \cA)$, we only show the tests of interest as follows.
\begin{alignat}{2}
\test(\cB(5, :) \circ \cA(3, :)) &= 0, \test(\cM \vee (\cB(5, :) \circ \cA(3, :))) &= [1, 0, 0, 1, 1, 1, 0, 0, 1]^T; \nonumber \\
\test(\cB(6, :) \circ \cA(3, :)) &= 0, \test(\cM \vee (\cB(6, :) \circ \cA(3, :))) &= [0, 0, 1, 1, 1, 1, 0, 0, 1]^T; \nonumber \\
\test(\cB(9, :) \circ \cA(3, :)) &= 0, \test(\cM \vee (\cB(9, :) \circ \cA(3, :))) &= [1, 0, 1, 0, 1, 1, 0, 0, 1]^T; \nonumber \\
\test(\cB(1, :) \circ \cA(4, :)) &= 0, \test(\cM \vee (\cB(1, :) \circ \cA(4, :))) &= [1, 0, 1, 0, 1, 1, 0, 0, 1]^T; \nonumber \\
\test(\cB(4, :) \circ \cA(4, :)) &= 0, \test(\cM \vee (\cB(4, :) \circ \cA(4, :))) &= [0, 0, 1, 1, 1, 1, 0, 0, 1]^T; \nonumber \\
\test(\cB(5, :) \circ \cA(4, :)) &= 0, \test(\cM \vee (\cB(5, :) \circ \cA(4, :))) &= [0, 0, 0, 0, 1, 1, 0, 0, 1]^T; \nonumber \\
\test(\cB(6, :) \circ \cA(4, :)) &= 0, \test(\cM \vee (\cB(6, :) \circ \cA(4, :))) &= [0, 0, 1, 1, 1, 1, 0, 0, 1]^T; \nonumber \\
\test(\cB(9, :) \circ \cA(4, :)) &= 0, \test(\cM \vee (\cB(9, :) \circ \cA(4, :))) &= [1, 0, 1, 0, 1, 1, 0, 0, 1]^T; \nonumber \\
\test(\cB(3, :) \circ \cA(5, :)) &= 0, \test(\cM \vee (\cB(3, :) \circ \cA(5, :))) &= [0, 0, 1, 1, 1, 1, 0, 0, 1]^T; \nonumber \\
\test(\cB(4, :) \circ \cA(5, :)) &= 0, \test(\cM \vee (\cB(4, :) \circ \cA(5, :))) &= [1, 0, 0, 1, 1, 1, 0, 0, 1]^T; \nonumber \\
\test(\cB(5, :) \circ \cA(5, :)) &= 0, \test(\cM \vee (\cB(5, :) \circ \cA(5, :))) &= [1, 0, 0, 1, 1, 1, 0, 0, 1]^T; \nonumber \\
\test(\cB(6, :) \circ \cA(5, :)) &= 0, \test(\cM \vee (\cB(6, :) \circ \cA(5, :))) &= [0, 0, 0, 0, 1, 1, 0, 0, 1]^T; \nonumber \\
\test(\cB(9, :) \circ \cA(5, :)) &= 0, \test(\cM \vee (\cB(9, :) \circ \cA(5, :))) &= [0, 0, 1, 1, 1, 1, 0, 0, 1]^T. \nonumber
\end{alignat}

\emph{Decoding:} Set $P = \emptyset$. For $i = 1, 2, 6$, we get $\test(\cA(1, :)) = \test(\cA(2, :)) = \test(\cA(6, :))= 0$. Therefore, we do not proceed to decode the test outcomes generated by rows $\cA(1, :), \cA(2, :)$,  and $\cA(6, :)$. From now, we examine rows $\cA(3, :), \cA(4, :)$, and $\cA(5, :)$. Consider row $\cA(3, :)$. Because $\test(\cB(1, :) \circ \cA(3, :)) = \test(\cB(2, :) \circ \cA(3, :)) = \test(\cB(3, :) \circ \cA(3, :)) = \test(\cB(4, :) \circ \cA(3, :)) = \test(\cB(7, :) \circ \cA(3, :)) = \test(\cB(8, :) \circ \cA(3, :)) = 1$, we only consider $\test(\cB(5, :) \circ \cA(3, :)) = \test(\cB(6, :) \circ \cA(3, :)) = \test(\cB(9, :) \circ \cA(3, :)) = 0$. We get $S_{5, 3} = \dec(\test(\cM \vee (\cB(5, :) \circ \cA(3, :))), \cM) = \{ 8, 11 \}$, $S_{6, 3} = \dec(\test(\cM \vee (\cB(6, :) \circ \cA(3, :))), \cM) = \{ 1, 11 \}$, and $S_{9, 3} = \dec(\test(\cM \vee (\cB(9, :) \circ \cA(3, :))), \cM) = \{ 1, 8 \}$. The defective set after processing row $\cA(3, :)$ and the tests generated by it is $P = P \cup S_{5, 3} \cup S_{6, 3} \cup S_{9, 3} = \{1, 8, 11 \}$.

Similarly, for row $\cA(4, :)$, we only consider $\test(\cB(1, :) \circ \cA(4, :)) = \test(\cB(4, :) \circ \cA(4, :)) = \test(\cB(5, :) \circ \cA(4, :)) = \test(\cB(6, :) \circ \cA(4, :)) = \test(\cB(9, :) \circ \cA(4, :)) = 0$. Then we get $S_{1, 4} = \{1, 8 \}, S_{4, 4} = \{1, 11 \}, S_{5, 4} = \emptyset, S_{6, 4} = \{1, 11 \}, S_{9, 4} = \{1, 8 \}$. The defective set after processing row $\cA(4, :)$ and the tests generated by it is $P = P \cup S_{1, 4} \cup S_{4, 4} \cup S_{5, 4} \cup S_{6, 4} \cup S_{9, 4} = \{1, 8, 11 \} \cup \{1, 8 \} \cup \{1, 11 \} \cup \emptyset \cup \{1, 11 \} \cup \{1, 8 \} = \{1, 8, 11 \}$. For row $\cA(5, :)$, we only consider $\test(\cB(3, :) \circ \cA(5, :)) = \test(\cB(4, :) \circ \cA(5, :)) = \test(\cB(5, :) \circ \cA(5, :)) = \test(\cB(6, :) \circ \cA(5, :)) = \test(\cB(9, :) \circ \cA(5, :)) = 0$. And therefore, we get $S_{3, 5} = \{1, 8 \}, S_{4, 5} = \{8, 11 \}, S_{5, 5} = \{8, 11 \}, S_{6, 5} = \emptyset, S_{9, 5} = \{1, 8 \}$. Consequently, the defective set after processing row $\cA(5, :)$ and the tests generated by it is $P = P \cup S_{3, 5} \cup S_{4, 5} \cup S_{5, 5} \cup S_{6, 5} \cup S_{9, 5} = \{1, 8, 11 \}$. This is also the final defective set.

\section{Non-adaptive and 2-stage designs for $u = 2$}
\label{sec:efficient}

In this section, we present designs with decoding time of $O(d^2 \log^2{(n/d)})$ to find the defective set for $(n, d, \ell = 1, u = 2)$-TGT. When $u = 2$, we only need two rows in which each row contains exactly one defective item and the two defective items in the two rows are distinct. This idea was proposed in~\cite{bui2023concomitant} for a randomized and non-adaptive design. Since our designs are deterministic, the technical details are different from theirs. We first present the notion of selectors which will be used for our designs later.

\subsection{Selectors}
\label{sec:pre:selectors}

To have an efficient decoding procedure, we recall a tool named \emph{selectors} which were introduced by De Bonis et al.~\cite{de2005optimal}.

\begin{definition}~\cite{de2005optimal}
Given integers $k$, $m$, and $n$, with $1 \leq m \leq k \leq n$, we say that a Boolean matrix M with $t$ rows and $n$ columns is a $(k, m, n)$-selector if any submatrix of $\cM$ obtained by choosing $k$ out of $n$ arbitrary columns of $\cM$ contains at least $m$ distinct rows of the identity matrix $\cI_k$. The integer $t$ is the size of the $(k, m, n)$-selector.
\label{def:selector}
\end{definition}


\begin{theorem}~\cite[Theorem 1]{de2005optimal}
Given integers $1 \leq m \leq k \leq n$, there exists a $(k, m, n)$-selector of size $t$, with
\[
t < \frac{\mathrm{e} k^2}{k - m + 1} \log{\frac{n}{k}} + \frac{\mathrm{e} k(2k - 1)}{k - m + 1}.
\]
\label{thm:selector}
\end{theorem}

\subsection{Main idea}
\label{sub:efficientDet}

Our main idea consists of two parts:
\begin{itemize}
\item Use a $(2d, d + 2, n)$-selector to obtain at least two indicating subsets in which each contains only one defective item and non-defective items.
\item For each indicating subset, add it to every test in SGT to find the defective items. In particular, a test is positive if at least one more defective item is present and negative otherwise. Hence, the standard group testing techniques can be applied to find the defective items.
\end{itemize}

We present the results of non-adaptive and 2-stage designs in Theorem~\ref{thm:u2Det} and their details in the subsequent sections. Note that the proofs of their complexities are presented in the full version.

\subsection{Non-adaptive design}
\label{sec:efficientDet:1Stage}

\textbf{Encoding:} We first initialize a $(2d, d + 2, n)$-selector $\cA$ of size $a \times n$ and an $m \times n$ $d$-disjunct matrix $\cM$. Then we generate three types of vectors: indicator vector $\bY := \test(\cA)$, reference vector $\bR := \test(\cA \vee \cA)$, and identification vector $\bF := \test(\cM \vee \cA)$. The indicator and reference vectors are used to locate two tests such that they contain two distinct defective items and each contains only one defective item. The identification vector is then to recover the defective set based on these two tests.

\textbf{Decoding:} We first find two indices $i_1 \neq i_2 \in [a]$ by scanning the indicator vector $\bY$ and the reference vector $\bR$ such that $\test(\cA(i_1, :)) = \test(\cA(i_2, :)) = 0$ and $\test(\cA(i_1, :) \vee \cA(i_2, :)) = 1$. Then the defective set is $P = \dec(\test(\cM \vee \cA(i_1, :)), \cM) \cup \dec(\test(\cM \vee \cA(i_2, :)), \cM)$.

\textbf{Correctness:} Consider an arbitrary set $C \subset [n]$ with $C \cap P = P$ and $|C| = 2d$. From Definition~\ref{def:selector}, since $\cA$ is a $(2d, d + 2, n)$-selector, there exists at least $d + 2$ distinct rows of the identity matrix $\cI_{2d}$. Because $|P| = d \geq 2$, there must exist two rows $i_1$ and $i_2$ in $\cA$ such that two distinct defective items belong to them and each row contains a single defective item. In other words, $\supp(\cA(i_1, :)) \cap P = \{ p_1 \}$ and $\supp(\cA(i_2, :)) \cap P = \{ p_2 \}$ for $p_1 \neq p_2 \in P$. Because $\test(\cA(i_1, :)) = \test(\cA(i_2, :)) = 0$, by Lemma~\ref{lem:CGTwithSet}, we get $\dec( \test(\cM \vee \cA(i_1, :)), \cM) = P \setminus \{ p_1 \}$ and $\dec( \test(\cM \vee \cA(i_2, :)), \cM) = P \setminus \{ p_2 \}$. Therefore, we always get $P \subseteq S$.

Now we prove that $S$ never contains false defective items. Indeed, because we only process row $i$ with $\test(\cA(i, :)) = 0$, row $\cA(i, :)$ satisfies $|\supp(\cA(i, :)) \cap P| \leq u - 1 = 1$. Therefore, by Lemma~\ref{lem:CGTwithSet}, the procedure $\dec( \test(\cM \vee \cA(i, :)), \cM)$ either return to a set of defective items or an empty set. Set $S$ thus never contains false defective items.

\textbf{Complexity:} By using the result in Theorem~\ref{thm:selector}, the number of tests in $\cA$ is $a = O \big( d \log{\frac{n}{d}} \big)$. By Theorem~\ref{thm:disjunct}, an $m \times n$ $d$-disjunct matrix $\cM$ can be efficiently decoded in $O \big( m + d \log^2{(n/d)} \big)$ time with $m = O \big( d^2 \min(n, d) \big)$. Therefore, the number of tests the non-adaptive design is:
\begin{align}
a + a^2 + am = O\left( d^3 \log{\frac{n}{d}} \cdot \min(n, d) \right). \nonumber
\end{align}

To identify $i_1$ and $i_2$, it takes $O(a^2) = O\big( d^2 \log^2{\frac{n}{d}} \big)$ time. Because of the construction of $\cM$, it takes $2 O\big( d^2 \min\{ \log{n}, (\log_d{n})^2 \} + d \log^2{(n/d)} \big)$ time to run $\dec(\test(\cM \vee \cA(i_1, :)), \cM)$ and $\dec(\test(\cM \vee \cA(i_2, :)), \cM)$. Therefore, it takes $O\left( d^2 \left(  \log^2{\frac{n}{d}} + \min(n, d) \right) \right)$ time to run the non-adaptive design.

\subsection{2-stage design}
\label{subsub:efficientDet:2Stage}

\textbf{Encoding and decoding:} In the first stage, we utilize matrices $\cA, \cM$, the indicator vector $\bY$, the reference vector $\bR$ in Section~\ref{sec:efficientDet:1Stage}. However, we amend vector $\bF$ as follows. We initialize a $(d, d, n)$-list disjunct matrix $\cH$ of size $h \times n$ as in Theorem~\ref{thm:listDisjunct} which can be decoded in $O\big( d \log^2{(n/d)} \big)$ time with $h = O\big( d \log{(n/d)} \big)$. Then we set $\bF = \test(\cH \vee \cA)$. 

In the second stage, we first find two indices $i_1 \neq i_2 \in [a]$ by scanning the indicator vector $\bY$ and $\bR$ such that $\test(\cA(i_1, :)) = \test(\cA(i_2, :)) = 0$ and $\test(\cA(i_1, :) \vee \cA(i_2, :)) = 1$. Then the approximate defective set is $R = \dec(\test(\cH \vee \cA(i_1, :)), \cH) \cup \dec(\test(\cH \vee \cA(i_2, :)), \cH)$. Finally, we test every two items in $R$. If a test on two items is positive then we add them to the defective set.

\textbf{Correctness:} Similar to the analysis in Section~\ref{sec:efficientDet:1Stage}, when $\test(\cA(i_1, :)) = \test(\cA(i_2, :)) = 0$ and $\test(\cA(i_1, :) \vee \cA(i_2, :)) = 1$, we get $|\cA(i_1, :) \cap P| = 1$ and $|\cA(i_2, :) \cap P| = 1$. Let $\{ p_1 \} = \cA(i_1, :) \cap P$ and $\{ p_2 \} = \cA(i_2, :) \cap P$. It is straightforward that $p_1 \neq p_2$.

When $|\supp(\bX) \cap P| = u - 1$, it is easy to confirm that $\cT \otimes \bX = \cT \odot \vecc(P \setminus (\supp(\bX) \cap P))$ for any matrix $\cT$. Therefore, we have $\test(\cH \vee \cA(i_1, :)) = \cH \odot \vecc(P \setminus (\cA(i_1, :) \cap P)) = \cH \odot \vecc(P \setminus \{ p_1 \})$ and $\test(\cH \vee \cA(i_2, :)) = \cH \odot \vecc(P \setminus (\cA(i_2, :) \cap P)) = \cH \odot \vecc(P \setminus \{ p_2 \})$. Since $\cH$ is a $(d, d)$-list disjunct matrix, we get $\dec(\test(\cH \vee \cA(i_k, :)), \cH) = P \setminus \{ p_k \}$ and $|\dec(\test(\cH \vee \cA(i_k, :)), \cH)| \leq 2d - 1$ for $k = 1, 2$. Therefore, $|R| \leq 3d - 2$. Since we are considering the case without gap with $u = 2$, if a test on two items is positive, the two items must be defective. On the other hand, since $R \cap P = P$, all pairs of the defective items will be tested, and therefore, by adding every two items in a test that yields a positive outcome will recover the defective set.

\textbf{Complexity:} The number of tests is:
\begin{align}
a + a^2 + a h + {|R| \choose 2} = O\left( d^2 \log^2{\frac{n}{d}} \right). \nonumber
\end{align}

In the second stage, to identify $i_1$ and $i_2$, it takes $O(a^2) = O\big( d^2 \log^2{\frac{n}{d}} \big)$ time. Because of the construction of $\cH$, it takes $O\big( d \log^2{(n/d)} \big)$ time to decode $\dec(\test(\cH \vee \cA(i_1, :)), \cH)$ and $ \dec(\test(\cH \vee \cA(i_2, :)), \cH)$. Finally, it takes $O\big( {|R| \choose 2} \big) = O({3d - 2 \choose 2}) = O(d^2)$ time to get the defective set $P$. In summary, it takes $O\big( d^2 \log^2{\frac{n}{d}} \big) + O\big( d \log^2{(n/d)} \big) + O(d^2) = O\big( d^2 \log^2{\frac{n}{d}} \big)$ time to run the 2-stage design.

\section{Conclusion}
\label{sec:cls}

In this paper, we first present efficient decoding designs with a polynomial time of $\min\{u^u, (d - u)^{d - u}\}, d$, and $\log{n}$ for TGT without gap under the deterministic design, the combinatorial setting, and the exact recovery criterion. It is inciting to extend this work for the case TGT with a gap. Another interesting question is whether the designs here are applicable to find hidden graphs in learning theory.

\bibliographystyle{ieeetr}
\balance
\bibliography{bibli}


\appendices

\section{Proof of Lemma~\ref{lem:CGTwithSet}}
\label{app:LemmaCGT}

Consider $y_i = \test(\cM(i, :) \vee \bX) = \test( \supp(\cM(i, :)) \cup \supp(\bX) )$ for $i = 1, \ldots, t$. Since $\cM$ is a $p$-disjunct matrix, For every defective item $r \in P$, there exists a row $i_r$, called \emph{a private row}, such that that row contains $r$ and does not contain the remaining items in $P \setminus \{r \}$. Therefore, if $|\supp(\bX) \cap P| < u - 1$, $y_{i_r} = 0$. Consequently, any decoding procedure with the exact recovery in standard group testing with the measurement matrix $\cM$ must declare item $r$ as non-defective. Therefore, the decoding procedure returns an empty set.

When $|\supp(\bX) \cap P| = u - 1$, for any $r \in \supp(\bX) \cap P$, we get $y_{i_r} = 0$ because row $i_r$ contains only $u - 1$ defective items in $\supp(\bX) \cap P$. Therefore, any item in $\supp(\bX) \cap P$ will not be included in the returned defective set by any decoding procedure with the exact recovery in standard group testing. On the other hand, for any $r^\prime \in P \setminus \supp(\bX)$, every test that includes $r^\prime$ is positive. Because $\cM$ is a $p$-disjunct matrix, item $i^\prime$ will be included in the returned defective set by any decoding procedure with the exact recovery in standard group testing. In other words, the returned defective set is $P \setminus \supp(\bX)$ if $|\supp(\bX) \cap P| = u - 1$.

\end{document}